\documentclass[preprint,12pt]{article}
\usepackage[utf8]{inputenc}
\usepackage{amsmath, amsthm, amssymb}
\usepackage{enumerate}
\usepackage{xcolor}
\usepackage{verbatim}
\usepackage{url}
\usepackage{geometry}
\usepackage{hyperref}
\newtheorem{theorem}{Theorem}
\newtheorem{lemma}[theorem]{Lemma}
\newtheorem{cor}[theorem]{Corollary}
\newtheorem{prop}[theorem]{Proposition}
\newtheorem{remark}[theorem]{Remark}

\newtheorem{example}[theorem]{Example}

\usepackage{amssymb}
\usepackage{amsmath}
\usepackage[normalem]{ulem}

\begin{document}

		



		\title{State Transfer and Readout Times for Trees of Diameter $4$}
		\author{Stephen Kirkland\textsuperscript{1}
			\thanks{Research supported in part by The Natural Sciences and Engineering Research Council of Canada (NSERC) Discovery Grant RGPIN--2019--05408. } \and 
			Christopher M. van Bommel\textsuperscript{1,2}
		}

		\date{\today}

		\addtocounter{footnote}{1}
		\footnotetext{Department of Mathematics, University of Manitoba, Winnipeg, MB, Canada.}
		\addtocounter{footnote}{1}
		\footnotetext{Department of Mathematics and Statistics, University of Guelph, Guelph, ON, Canada. \\   Stephen.Kirkland@umanitoba.ca, cvanbomm@uoguelph.ca  }

		\maketitle

		%% Abstract
		\begin{abstract}
			We consider the state transfer properties of continuous time quantum walks on trees of diameter 4. We characterize all pairs of strongly cospectral vertices in trees of diameter 4, finding that they fall into pairs of three different types. For each type, we construct an infinite family of diameter 4 trees for which there is pretty good state transfer between the pair of strongly cospectral vertices. Moreover,  for two of those types, for each tree in the infinite family, we give an explicit sequence of readout times at which the fidelity of state transfer converges to $1$. For strongly cospectral vertices of the remaining type, we identify a sequence of trees and explicit readout times so that the fidelity of state transfer between the strongly cospectral vertices approaches $1.$   
			
			We also prove a result of independent interest: for a graph with the property that the fidelity of state transfer between a pair of vertices at time  $t_k$ converges to $1$ as $k \rightarrow \infty,$ then the derivative of the fidelity at $t_k$ converges to $0$ as $k \rightarrow \infty. $ 
		\end{abstract}



		
		
		\section{Introduction}
		
		Interacting qubits can be modelled by a graph, where the individual qubits are represented by vertices of the graph and the presence of interactions between qubits is represented by an edge between the corresponding vertices.  This forms a natural extension of classical random walks to the quantum setting, which we refer to as \emph{continuous quantum walks}.  We will consider the propagation model of this quantum system corresponding to the adjacency matrix $A$ of the corresponding graph; then the time evolution of the system is given by
		$
		U(t) = \exp(i t A),
		$
		where $t$ is a nonnegative real value representing the length of time under which the system has propagated.  It follows that $U(t)$ is both symmetric and unitary, and in particular that  $|U({t})_{a, b}| \le  1$ for any $t$.  The \emph{fidelity} at time $t$ is defined as $|U({t})_{a, b}|^2,$ and is interpreted as the probability of state transfer from $a$ to $b$ at time $t$. 
		We say a graph has \emph{perfect state transfer} between vertices $a$ and $b$ at time $\hat{t}$ if $|U(\hat{t})_{a, b}| = 1$.  
		
		Perfect state transfer is an important task in the area of quantum information processing, yet the graphs in which this property can be found appear to be rare.  Christandl et al.~\cite{CDEL04, CDDEKL05} demonstrated that the path graph has perfect state transfer between its end vertices if and only if the length of the path is 1 or 2; Stevanovi\'c~\cite{S11} and Godsil~\cite{G12} independently extended this result to exclude perfect state transfer between internal nodes of paths.  Godsil~\cite{G12b} found that for each $k \in \mathbb{N},$ the number of connected graphs with valency at most $k$ that admit perfect state transfer is finite. Coutinho and Liu~\cite{CL14} considered the quantum system corresponding to the Laplacian matrix of the graph, and demonstrated that the path on $2$ vertices is the only tree that admits perfect state transfer in this model.  Recently, Coutinho, Juliano, and Spier~\cite{CJS23} demonstrated that the only trees that admit perfect state transfer with respect to the adjacency matrix are $P_2$ and $P_3$.

		In light of the scarcity of small examples of perfect state transfer, we make a small sacrifice to the fidelity of the transfer and consider \emph{pretty good state transfer} (PGST), which requires a sequence of times $t_k$ such that $\lim_{k \to \infty} |U(t_k)_{a, b}| = 1$.  Pretty good state transfer between end-vertices of paths was studied by Godsil, Kirkland, Severini, and Smith~\cite{GKSS}, and determined to be achieved precisely on paths whose number of vertices was one less than one of the following: a prime, twice a prime, or a power of $2$.  Coutinho, Guo, and van Bommel~\cite{CGvB} provided a family of paths in which pretty good state transfer was achieved between internal vertices, and van Bommel~\cite{vB19} subsequently completed the characterization by demonstrating this was the only such family.  %These results can be summarized as follows.
		%
		\begin{comment}
			\begin{theorem}\label{thm:pn}
				There is pretty good state transfer on $P_n$ between vertices $a$ and $b$ if and only if $a + b = n + 1$ and:
				\begin{enumerate}[a)]
					\item $n = 2^t - 1$, where $t$ is a positive integer;
					\item $n = p - 1$, where $p$ is an odd prime; or
					\item $n = 2^t p - 1$, where $t$ is a positive integer and $p$ is an odd prime, and $a$ is a multiple of $2^{t - 1}$.
				\end{enumerate}
			\end{theorem}
		\end{comment}
		Fan and Godsil~\cite{FG13} studied the double star graph $S_{k, \ell}$, i.e., the graph formed by joining two stars $K_{k, 1}$ and $K_{\ell, 1}$ by an edge between the centre vertices, and demonstrated that pretty good state transfer occurs precisely between the stems or leaves of $P_4$, the leaves of the $K_{1, 2}$ in the graph $S_{2, \ell}$, $\ell \neq 2$, or the centre vertices of $S_{k, k}$ when $1 + 4k$ is not a perfect square.
		%
		\begin{comment}
			\begin{theorem} \label{thm:pgst-double-star} \cite{FG13} 
				Pretty good state transfer is admitted between vertices $u$ and $v$ of the double star graph $S_{k, \ell}$ if and only if:
				\begin{enumerate}[a)]
					\item $S_{k, \ell} \cong P_4$, and $u$ and $v$ are the two leaves or the two stems, or
					\item $k = 2$, $\ell \neq 2$, and $u$ and $v$ are the two leaves corresponding to $K_{1, 2}$, or
					\item $k = \ell > 2$, $1 + 4k$ is not a perfect square, and $u$ and $v$ are the two centre vertices.
				\end{enumerate}
			\end{theorem}
		\end{comment}
		%
		Hou, Gu, and Tong~\cite{HGT18} considered extended double stars $F_{k, \ell}$, which are constructed from two stars $K_{k + 1, 1}$ and $K_{\ell + 1, 1}$ by taking a leaf from each and identifying those vertices, and proved pretty good state transfer occurs precisely for stems or leaves of $P_5$ or the stems of $F_{k, k}$.
		%
		\begin{comment}
			\begin{theorem} \label{thm:pgst-eds} \cite{HGT18}
				There is pretty good state transfer on $F_{k,\ell}$ if and only if:
				\begin{enumerate}[a)]
					\item $k = \ell = 1$ (i.e. $P_5$), and pretty good state transfer occurs between the two end-vertices or the second and fourth vertices; or
					
					\item $k = \ell > 1$, and pretty good state transfer occurs between the two original centre vertices.
				\end{enumerate}
			\end{theorem}
		\end{comment}
		%
		We observe in passing that the double stars are precisely the trees of diameter $3$ while the extended double stars are a subfamily of the trees of diameter $4$. 
		
		The present paper considers pretty good state transfer on trees, with a focus on trees of diameter $4$.  While many results on PGST characterize circumstances under which that phenomena takes place, there is  scant literature that addresses the issue of readout times. 
		One of our novel contributions in the present work is, for certain families of trees, the explicit description of readout times at which the fidelity of state transfer is close to $1$. By so doing we enhance the applicability of PGST. We also prove that for any graph, if $t_k$ is a sequence of times such that $|U({t_k})_{a, b}| \rightarrow 1$ as $k \rightarrow \infty,$ then the derivative of the fidelity at time $t_k$ converges to $0$ as $k \rightarrow \infty.$ Consequently, for a graph exhibiting PGST between $a$ and $b$ the fidelity at $t_k$ becomes less sensitive to small changes in the readout time as $k \rightarrow \infty.$ 
		
		%Some trees with small diameter have been further investigated, and perfect state transfer has been ruled out in these cases.  It is straightforward to establish there can be no perfect state transfer in stars on at least four vertices; we provide a proof after introducing the necessary background.  

		\section{Preliminaries}
		
		We analyze the time evolution of quantum systems through the spectral decomposition of the adjacency matrix, which is given by
		$
		A = \sum_{\lambda} \lambda E_\lambda,
		$
		where we sum over the eigenvalues $\lambda$ of $A$ and $E_\lambda$ is the projection onto the $\lambda$-eigenspace.  %(When it is important to consider their ordering, we will denote the eigenvalues of $A$  by $\lambda_1 > \lambda_2 > \cdots > \lambda_m$.)  
		The time evolution of the system can be expressed as
		$
		U(t) = \sum_{\lambda} \exp(i t \lambda) E_\lambda.
		$
		The perfect state transfer condition of $|U(\hat{t})_{a, b}| = 1$ is equivalent to $U(\hat{t}) e_a = \gamma e_b$ for some $\gamma \in \mathbb{C}$ with $|\gamma| = 1$, where $e_j$ denotes the standard basis vector which contains a 1 in row $j$ and 0 in all other entries.  Equivalently, we have $e^{i t \lambda} E_\lambda e_a = \gamma E_\lambda e_b$ for each eigenvalue $\lambda$, which is equivalent to the following pair of properties:
		\begin{enumerate}[(a)]
			\item For each eigenvalue $\lambda$, there is a $\sigma_\lambda \in \{-1, +1\}$ such that $E_{\lambda} e_a = \sigma_\lambda E_\lambda e_b$,  and
			\item For each eigenvalue $\lambda_r$, whenever $E_{\lambda_r} e_a \neq 0$, then $t(\lambda_1 - \lambda_r) = k_r \pi$, where $k_r \in \mathbb{Z}$, and $k_r \equiv (1 - \sigma_{\lambda_r}) / 2 \pmod{2}$.
		\end{enumerate}
		If vertices $a$ and $b$ satisfy  property (a) above, then we say they are \emph{strongly cospectral}.  We note that this condition implies the more well-known condition of being \emph{cospectral}, which requires that $(E_\lambda)_{a, a} = (E_\lambda)_{b, b}$ for all eigenvalues $\lambda$.  It follows that such vertices satisfy $A^k_{a, a} = A^k_{b, b}$ for every integer $k$, which counts the number of closed walks from each vertex.  In particular, we have that $a$ and $b$ necessarily have the same degree.  We distinguish the eigenvalues for which $E_\lambda e_a \neq 0$ as forming the \emph{eigenvalue support} of $a$, denoted $\Theta_a$, since these are the only terms that affect the time evolution of the system at that vertex. Observe that if $a,b$ are strongly cospectral vertices then $\Theta_a$ can be partitioned naturally as $\Theta_a= \sigma_{ab}^+ \cup \sigma_{ab}^-,$ 
		where $\sigma_{ab}^+=\{\lambda \in \Theta_a| E_{\lambda} e_a =  E_\lambda e_b \}$ and $\sigma_{ab}^-=\{\lambda \in \Theta_a| E_{\lambda} e_a =  -E_\lambda e_b \}.$

		We note the following useful consequence of strong cospectrality, due to Godsil and Smith~\cite{GS17}.
		
		\begin{lemma} \cite{GS17} \label{lem:auto}
			If $a$ and $b$ are strongly cospectral vertices in a graph $X$, then any automorphism of $X$ that fixes $a$ also fixes $b$.
		\end{lemma}
		
		% ?For pretty good state transfer, we have the following characterization provided by Banchi et al.~\cite{BCGS}.
		
		\begin{comment}
			\begin{theorem} \cite{BCGS}
				Let $a$ and $b$ be vertices of a graph $G$.  Then pretty good state transfer occurs between $a$ and $b$ if and only if both conditions below are satisfied:
				\begin{enumerate}[a)]
					\item Vertices $a$ and $b$ are strongly cospectral, in which case for each $\lambda_r \in \Theta_a,$ let $k_r \equiv (1 - \sigma_{\lambda_r}) / 2 \pmod{2}$.
					\item If there is a set of integers $\{\ell_j\}$ such that
					\[ \sum_{\lambda_j \in \Theta_a} \ell_j \lambda_j = 0 \text{ and } \sum_{\lambda_j \in \Theta_a} \ell_j k_j \text{ is odd,} \]
					then necessarily 
					\[ \sum_{\lambda_j \in \Theta_a} \ell_j \neq 0.\]
				\end{enumerate}
			\end{theorem}
		\end{comment}
		
		Finally, we note that for trees, Coutinho, Juliano, and Spier demonstrated that a vertex can have at most one other vertex to which it is strongly cospectral~\cite{CJS22}.

		%We now see why perfect state transfer is ruled out for stars on at least four vertices.  We first observe that the centre vertex has degree at least three, and therefore does not have the same degree as any other vertex, so it cannot be part of a strongly cospectral pair.  Now, suppose by way of contradiction that $a$ and $b$ are cospectral leaves, and let $c$ be another leaf of the graph.  Then the automorphism swapping vertices $b$ and $c$ and fixing all other vertices is an automorphism that fixes $a$ but does not fix $b$, contradicting that $a$ and $b$ are strongly cospectral.  Therefore, no star on at least four vertices admits perfect state transfer.

		\section{Sensitivity of state transfer for strongly cospectral vertices}
		
		If there is perfect state transfer from one state to another at time $\hat t,$ then necessarily the derivative of the fidelity is zero at $\hat t$, since $\hat t$ corresponds to a local maximum for the fidelity. In this section, we obtain an expression for the derivative of the fidelity between strongly cospectral vertices, with particular interest in pairs of vertices exhibiting PGST. Observe that our result is not restricted to trees. 
		
		\begin{theorem}\label{thm:sens}
			Suppose that $a,b$ are strongly cospectral vertices in a graph on $n$ vertices, and denote the fidelity at time $t$ by $f(t)$. For each eigenvalue $\lambda$ of $A$, let $s_\lambda$ denote the $(a,a)$ entry of the corresponding orthogonal idempotent projection matrix $E_\lambda$.  
			% Let $V$ denote an orthogonal eigenmatrix, where for %each $j=1, \ldots, n,$ the $j$--th column of $V$ %corresponds to the eigenvalue $\lambda_j$. 
			Then
			\begin{eqnarray}\label{eq:sens} \nonumber 
				&&\frac{df}{dt} = 
				2 \left[\sum_{\lambda_j, \lambda_\ell \in \sigma_{ab}^+}s_{\lambda_j}s_{\lambda_\ell}\lambda_\ell \sin(t(\lambda_j-\lambda_\ell))+
				\sum_{\lambda_j, \lambda_\ell \in \sigma_{ab}^-}s_{\lambda_j}s_{\lambda_\ell}\lambda_\ell \sin(t(\lambda_j-\lambda_\ell))
				+ \right.\\ 
				&&\left.\sum_{\lambda_j \in \sigma_{ab}^+, \lambda_\ell \in \sigma_{ab}^-}s_{\lambda_j}s_{\lambda_\ell}(\lambda_j-\lambda_\ell) \sin(t(\lambda_j-\lambda_\ell))\right].
			\end{eqnarray}
		\end{theorem}
		\begin{proof}
			We have 
			\begin{eqnarray*}
				&&f(t)=\\
				&&\left(\sum_{\lambda_j \in \sigma_{ab}^+}s_{\lambda_j} \cos(t\lambda_j)-
				\sum_{\lambda_j \in \sigma_{ab}^-}s_{\lambda_j} \cos(t\lambda_j)\right)^2 + \left(\sum_{\lambda_j \in \sigma_{ab}^+}s_{\lambda_j} \sin(t\lambda_j)-
				\sum_{\lambda_j \in \sigma_{ab}^-}s_{\lambda_j}\sin(t\lambda_j)\right)^2.
			\end{eqnarray*}
			Differentiating with respect to $t$, collecting terms, and simplifying, we obtain the desired result.
		\end{proof}

		\begin{example} {\rm{
					In this example, we illustrate the conclusion of Theorem \ref{thm:sens} for $P_3,$ where $a$ and $b$ are taken to be leaves. We have $\sigma_{ab}^+=\{\sqrt{2}, -\sqrt{2}\}, \sigma_{ab}^-=\{0\},$ with $s_{\sqrt{2}}=s_{-\sqrt{2}} = \frac{1}{4}$ and $ s_0=\frac{1}{2}.$ We find readily that $f(t)=\frac{1}{4}(1-\cos(\sqrt{2} t))^2,$ and differentiating that expression with respect to $t$ yields $$\frac{df}{dt} = \frac{1}{\sqrt{2}}(1-\cos(\sqrt{2} t))\sin(\sqrt{2} t).$$ 
					
					Alternatively, we can find $\frac{df}{dt}$ via \eqref{eq:sens} as follows: 
					\begin{eqnarray*}
						&&\frac{df}{dt}=\\
						&&2\left(\frac{1}{4}\frac{1}{4} \sqrt{2} \sin(-2\sqrt{2} t) + \frac{1}{4}\frac{1}{4}(- \sqrt{2}) \sin(2\sqrt{2} t)  + \right. \\
						&&\left. \frac{1}{4}\frac{1}{2} \sqrt{2} \sin(\sqrt{2} t) + \frac{1}{4}\frac{1}{2}(- \sqrt{2}) \sin(-\sqrt{2} t) \right) =\\
						%&&2\left(\frac{\sqrt{2}}{4}\sin(\sqrt{2} t) -\frac{\sqrt{2}}{8}\sin(2\sqrt{2} t)\right) =\\
						%&&\frac{\sqrt{2}}{2}(\sin(\sqrt{2} t)- \sin(\sqrt{2} t)\cos(\sqrt{2} t))=
						&&\frac{1}{\sqrt{2}}(1-\cos(\sqrt{2} t))\sin(\sqrt{2} t),
					\end{eqnarray*}
					as expected. }}
		\end{example} 
		
		The following corollary provides useful information in the context of PGST. 
		
		\begin{cor}\label{cor:sens} Suppose that there is PGST from $a$ to $b$. Denote the fidelity from $a$ to $b$ at time $t$ by $f(t),$ and let $\tau_m, m\in \mathbb{N}$ denote a sequence of times such that $f(\tau_m)\rightarrow 1$ as $m \rightarrow \infty.$ Then  $\frac{df}{dt} \Big |_{\tau_m} \rightarrow 0$ as $m \rightarrow \infty.$ 
		\end{cor}
		\begin{proof}
			Since there is PGST from $a$ to $b$, there is a sequence $\gamma_m \in \mathbb{C}$ such that: i) $|\gamma_m|=1, m \in \mathbb{N}$; 
			and ii) $\sum_{\lambda \in \Theta_a}e^{i\tau_m \lambda}E_\lambda e_a-\gamma_m e_b \rightarrow 0$ as $m \rightarrow \infty.$ 
			(This observation is essentially made  in \cite{BCGS}.) 
			Writing $e_b = \sum_{\lambda \in \Theta_a}E_\lambda e_b$, we may rewrite ii) as $\sum_{\lambda \in \sigma_{ab}^+}
			(e^{i\tau_m \lambda}-\gamma_m) E_\lambda e_a + \sum_{\lambda \in \sigma_{ab}^-}
			(e^{i\tau_m \lambda}+\gamma_m) E_\lambda e_a \rightarrow 0
			$ as $m \rightarrow \infty.$ 
			
			We now deduce that:  
			iii) for each $\lambda \in \sigma_{ab}^+, e^{i\tau_m \lambda}-\gamma_m \rightarrow 0$ 
			as $m  \rightarrow \infty$; and 
			iv) for each $\lambda \in \sigma_{ab}^-, e^{i\tau_m \lambda}+\gamma_m \rightarrow 0$ as $m  \rightarrow \infty$. Hence, we find that if $\lambda_j, \lambda_\ell \in \sigma_{ab}^+, $ then $e^{i\tau_m (\lambda_j-\lambda_\ell)}\rightarrow 1$ as $m  \rightarrow \infty$. Similarly, 
			$e^{i\tau_m (\lambda_j-\lambda_\ell)}\rightarrow 1$ as $m  \rightarrow \infty$ whenever  $\lambda_j, \lambda_\ell \in \sigma_{ab}^-,$ while $e^{i\tau_m (\lambda_j-\lambda_\ell)}\rightarrow -1$ as $m  \rightarrow \infty$ whenever  $\lambda_j \in   \sigma_{ab}^+, \lambda_\ell \in \sigma_{ab}^-. $  In particular, as $m  \rightarrow \infty,$ we have $\sin(\tau_m (\lambda_j-\lambda_\ell))\rightarrow 0$ if either $\lambda_j, \lambda_\ell \in \sigma_{ab}^+,$ or  $\lambda_j, \lambda_\ell \in \sigma_{ab}^-,$ or 
			$\lambda_j \in   \sigma_{ab}^+, \lambda_\ell \in \sigma_{ab}^-. $ The conclusion now follows from Theorem \ref{thm:sens}. 
		\end{proof}

		\section{Algebraic \& combinatorial structure of trees of diameter 4}
		\label{sec:struc}
		
		We now consider the properties of trees of diameter $4$ that will allow us to consider PGST and readout times.  We first observe that if we delete the leaves of a tree with diameter 4, we obtain a star on at least $3$ vertices.  This allows us to describe the structure of a tree with diameter 4 as a centre vertex (which we denote by $v$ henceforth) with some number of adjacent stems, each of which has a number of adjacent leaves.  (For ease of notation, we will consider a leaf of the centre vertex to be a stem with 0 adjacent leaves.)  We begin this section by classifying the pairs of strongly cospectral vertices of trees of diameter 4, as this is a necessary condition for PGST.  We note that condition (b) below is an example of the construction of strongly cospectral vertices  due to Godsil and Smith \cite[Theorem 9.1]{GS17}, while (c) is the ``rabbit-ear'' construction of strongly cospectral vertices in \cite[Lemma 9.2]{GS17}.
		
		\begin{theorem}\label{thm:str_cospec}
			Let $T$ be a tree of diameter 4 described as follows: there is a $t \in \mathbb{N}$ and parameters $q_j, a_j, j=1, \ldots, t$ such that $q_1 < q_2 < \cdots < q_t$, and for each $j=1, \ldots, t$ there are $a_j$ stems that are adjacent to the central vertex $v$, and also adjacent to $q_j$ leaves. 
			%    consisting of a central vertex $v$ adjacent to $a_i$ vertices each with $q_i$ leaves, where $q_1 < q_2 < \cdots < q_t$ for some integer $t$.  
			Then vertices $a$ and $b$ of $T$ are strongly cospectral if and only if one of the following holds:
			\begin{enumerate}[(a)]
				\item $a$ and $b$ are the only neighbours of $v$ with $q_j$ leaves (i.e. $a_j = 2$),
				\item $a$ and $b$ are the only leaves at distance 2 from $v$ with no twin (i.e. $q_j = 1, a_j = 2$),
				\item $a$ and $b$ are the only leaves of some stem, and the central vertex $v$ has at least one leaf (i.e. $q_1 = 0$, $q_j = 2$).
			\end{enumerate}
		\end{theorem}
		
		\begin{proof}
			Assume $a$ and $b$ are strongly cospectral.  Then $\deg(a) = \deg(b)$ and, since $A_{a, a}^4 = A_{b, b}^4$, and $T$ has no 4-cycles, we find that $a$ and $b$ have the same number of vertices at distance 2.  Suppose $a$ is a neighbour of $v$ so that the number of vertices at distance 2 from $a$ is $\sum_{j=1}^ta_j -1$, while $\deg(a)=q_{j_0}+1$ for some $j_0$.  If $b = v$, then the number of vertices at distance 2 from $b$ is $\sum_{j=1}^t a_jq_j$ and $\deg(b)=\sum_{j=1}^ta_j .$ Suppose now that $\sum_{j=1}^ta_j -1=\sum_{j=1}^t a_jq_j$ and $q_{j_0}+1=\sum_{j=1}^ta_j.$ From the first equation we find that $\sum_{j=1}^ta_j(q_j-1)=-1,$ so that necessarily $q_1=0$ and $a_1=\sum_{j=2}^ta_j(q_j-1)+1.$ But then the second equation yields $q_{j_0}+1 = a_1 +\sum_{j=2}^ta_j =\sum_{j=2}^ta_jq_j+1. $ 
				%\sout{From this we deduce that $q_{j_0}\ge 1$} 
				Note that since $T$ is a tree of diameter 4, $q_{j_0} + 1 = \deg(v) \ge 2$ (i.e. $j_0\ge 2$) and hence that $q_{j_0}+1 \ge a_{j_0}q_{j_0}+1,$ with strict inequality if $t\ge 3.$ We conclude that necessarily $t=2,$ in which case $T$ has diameter 3, contrary to our hypothesis.

			%the only vertices at distance 2 from $a$ are the other neighbours of $b$, so the only vertices at distance 2 from $b$ are the other neighbours of $a$.  Therefore, $T$ is a double star, contradicting the hypothesis that $T$ has diameter 4.
			
			Suppose $b$ is distance 2 from $v$, in particular, $b$ is a leaf.  Then $a$ is also a leaf.  Let $w$ be the neighbour of $b$, then $\deg(w) = \deg(v)$.  If $\deg(v) = 2$, then $T$ is a double star, contracting that $T$ has diameter 4.  If $\deg(v) \ge 3$, then $w$ is adjacent to another leaf, and there is an automorphism swapping this vertex and $b$ fixing all other vertices, so by Lemma~\ref{lem:auto}, $a$ and $b$ are not strongly cospectral, a contradiction.
			
			Finally, suppose $b$ is a neighbour of $v$.  If there exists a third neighbour of $v$ with the same number of leaves, then there is an automorphism swapping this vertex and its leaves with $y$ and its leaves fixing all other vertices, so by Lemma~\ref{lem:auto}, $a$ and $b$ are not strongly cospectral, a contradiction.  Hence, if $a$ is a neighbour of $v$, then (a) holds.
			
			Now, suppose $a$ is a leaf of $T$ that is not a neighbour of $v$.  Then $b$ is also a leaf of $T,$ and from the argument given above we may assume that $b$  is not a neighbour of $v$.  Since $T$ has diameter 4, we have $v \neq a, b$.  If $a$ and $b$ do not have a common neighbour, then the unique neighbours of $a$ and $b$ must have degree 2, otherwise there is an automorphism swapping $a$ or $b$ with its twin and fixing all other vertices, so by Lemma~\ref{lem:auto}, $a$ and $b$ are not strongly cospectral, a contradiction.  Moreover, if there is a third neighbour of $v$ with degree 2, then there is an automorphism swapping this stem-leaf pair with $b$ and its stem and fixing all other vertices, so by Lemma~\ref{lem:auto}, $a$ and $b$ are not strongly cospectral, a contradiction.  Hence, if $a$ is a leaf with no twin, then (b) holds.
			
			Finally, if $a$ and $b$ have a common neighbour $w$, then $\deg(w) = 3$, otherwise there exists a leaf $z \neq a, b$ adjacent to $w$, and then there is an automorphism swapping $b$ and $z$ and fixing all other vertices, so by Lemma~\ref{lem:auto}, $a$ and $b$ are not strongly cospectral, a contradiction.  Suppose $v$ is adjacent to no leaves.  We observe $T$ has eigenvectors $\alpha$ with $\alpha(a) = 1$, $\alpha(b) = -1$, and zero otherwise; and $\beta$ with $\beta(v) = -1$, $\beta(\ell) = \frac{1}{q_i}$ for each leaf $\ell$ of a stem with $q_i$ leaves, and zero otherwise.  Then $\alpha$ is an eigenvector corresponding to eigenvalue 0 for which $\alpha(a) = - \alpha(b)$, and $\beta$ is an eigenvector corresponding to eigenvalue 0 for which $\beta(a) = \beta(b)$, which contradicts strong cospectrality.  Hence, if $a$ is a leaf with a twin, then (c) holds.
			
			Conversely, we demonstrate that each case is strongly cospectral.  Suppose $a$ and $b$ are not strongly cospectral. In each case the vertices $a, b$ are clearly cospectral, so if they are not strongly cospectral then by \cite[Lemma 7.1]{GS17}, there is an eigenvalue $\lambda$ such that $(e_a^TE_\lambda e_a)(e_b^TE_\lambda e_b) -(e_a^TE_\lambda e_b)^2 \ne 0.$ Now we consider the $\lambda$-eigenvector $$\alpha = \frac{1}{(e_a^TE_\lambda e_a)(e_b^TE_\lambda e_b) -(e_a^TE_\lambda e_b)^2}((e_a^TE_\lambda e_a)(E_\lambda)e_b - (e_a^TE_\lambda e_b)(E_\lambda)e_a)
				$$ 
				%\sout{  Without loss of generality, there exists an eigenvector} 
				and observe that 
				%\sout{$\alpha$ with} 
				$\alpha(a) = 0$ and $\alpha(b) = 1.$ 
			%\sout{ corresponding to some eigenvalue $\lambda$.}
			
			\begin{enumerate}[(a)]
				\item Suppose $a$ and $b$ are the only neighbours of $v$ with $q_j$ leaves.  First suppose $a$ and $b$ are leaves.  It follows that $\alpha(v) = \lambda = 0$.  
				Observe that if $c$ is a leaf adjacent to stem $d$, then $\alpha(d)=\lambda \alpha(c)=0,$ and hence  
				%\sout{  Moreover,}
				$\alpha$ must be 0 on every stem of $T$.  But then $b$ is the only neighbour of $v$ whose corresponding entry in $\alpha$ is nonzero. Hence $(A\alpha)(v) = \alpha(b) =  1 \neq 0$, a contradiction.  Therefore, there is no such eigenvector, so $a$ and $b$ are strongly cospectral.  Now suppose $a$ and $b$ are not leaves.  Considering a leaf of $a$, we must have $\lambda = 0$, but considering a leaf of $b$, we must have $\lambda \neq 0$, a contradiction.  Therefore, there is no such eigenvector, so $a$ and $b$ are strongly cospectral.
				
				\item Suppose $a$ and $b$ are the only leaves at distance 2 from $v$ with no twin.  Considering $a$, we have that $\alpha$ assigns its stem $0$ and so $\alpha(v) = 0$.  Considering $b$, we have that $\alpha$ assigns its stem $\lambda$ and so $\alpha(v) = \lambda^2 - 1$.  It follows that $\lambda = \pm 1$.  Considering $v$, there must be a neighbour $z$ of $v$ with nonzero entry $c$.  Then each other neighbour of $z$ has entry $\pm \frac{c}{q_j} = \frac{c}{\lambda}$, so $q_j = 1$, contradicting that $a$ and $b$ are the only leaves with no twin.  Therefore, there is no such eigenvector, so $a$ and $b$ are strongly cospectral.
				
				\item Suppose $a$ and $b$ are the only leaves of some stem, and the central vertex has at least one leaf.  Let $s$ be their common neighbour.  It follows that $\alpha(s) = \lambda = 0$ and $\alpha(v) = -1$.  Since every stem is assigned 0 by $\alpha$, there is a leaf $\ell$ of $v$ with nonzero entry, but $(A\alpha)(\ell) = -1 \neq 0$.  Therefore, there is no such eigenvector, so $a$ and $b$ are strongly cospectral.  
			\end{enumerate}
		\end{proof}

		%In this section, we demonstrate that no tree with diameter 4 admits perfect state transfer.  We begin with the following lemma to analyze the eigenvalues of such a tree.
		
		We now analyze the eigenvalues of a tree of diameter 4.
		
		\begin{lemma} \label{lem:leaf-interlace}
			Let $T$ be a tree of diameter 4 consisting of a central vertex $v$ adjacent to $a_j$ vertices each with $q_j$ leaves, where $q_1 < q_2 < \cdots < q_t$ for some integer $t$.  Then the roots of the equation
			\begin{equation}\label{eq:sec}
				\sum_{j = 1}^t \frac{a_j}{x^2 - q_j} = 1
			\end{equation}
			are the nonzero eigenvalues of $T$ in the support of $v$.  Moreover, the squares of those nonzero eigenvalues satisfy
			$
			0 \le q_1 < \lambda_1^2 < q_2 < \lambda_2^2 < \cdots < q_t < \lambda_t^2.
			$
		\end{lemma}
		
		\begin{proof}
			Let $w$ be an eigenvector of $T$ in the support of $v$ with eigenvalue $x$.  Without loss of generality, we may assume $w(v) = 1$.  Let $s$ be a neighbour of $v$ with $q_j$ leaves.  Then for each leaf $\ell$ of $s$, we have $x w(\ell) = w(s)$ and $x w(s) = q_j w(\ell) + 1$.  It follows that $w(\ell) = \frac{1}{x^2 - q_j}$ and $w(s) = \frac{x}{x^2 - q_j}$.  Finally, we must have $x = x w(v) = \sum_{s \sim v} w(s) = \sum_{s \sim v} \frac{x}{x^2 - q_j}$.  Dividing through by $x$ and grouping vertices with the same number of leaves gives the desired relation. We note in passing that from the argument above, the eigenvector $w$ is unique up to scalar multiple. 
			
			Conversely, if $\lambda$ satisfies $\sum_{j = 1}^t \frac{a_j}{x^2 - q_j} = 1$, then the vector $w$ given by $w(v) = 1$, $w(s) = \frac{\lambda}{\lambda^2 - q_j}$ for each neighbour $s$ of $v$ with $q_j$ leaves, and $w(\ell) = \frac{1}{\lambda^2 - q_j}$ for each leaf $\ell$ adjacent to $s$ is an eigenvector with eigenvalue $\lambda$ in the support of $v$.
			
			{Fix $\ell \in \{1, \ldots, t-1\}$ and observe that as $x^2 \rightarrow q_\ell^+, \sum_{j = 1}^t \frac{a_j}{x^2 - q_j} \rightarrow \infty, $ while as $x^2 \rightarrow q_{\ell +1}^-, \sum_{j = 1}^t \frac{a_j}{x^2 - q_j} \rightarrow -\infty.$ }
			It follows from the intermediate value theorem that for each $\ell=1, \ldots, t-1$ there is a root $\lambda$ of  \eqref{eq:sec} such that  $q_\ell < \lambda^2 <  q_{\ell + 1},$ and similarly there is a root $\lambda $ such that $q_t < \lambda^2.$ 
			Now, suppose there exist eigenvalues $\lambda, \mu$ such that $q_\ell < \lambda^2 < \mu^2 < q_{\ell + 1}$ for some $\ell \le t$.  Then we have
			\[
			1 = \sum_{j = 1}^t \frac{a_j}{\lambda^2 - q_j} > \sum_{j = 1}^{t} \frac{a_j}{\mu^2 - q_j} = 1,
			\]
			a contradiction.  Moreover, if $\lambda^2 < q_1$, we have
			\[
			1 = \sum_{j = 1}^t \frac{a_j}{\lambda^2 - q_j} < 0,
			\]
			a contradiction.  The result follows.
		\end{proof}
		
		Working with \eqref{eq:sec} for each of the eigenvalues in the support of $v$, we consider fixed eigenvalues and numbers of leaves, and solve for the parameters counting the number of stems with a given number of leaves.
		
		\begin{lemma} \label{lem:Cauchy}
			Let $q_1, q_2, \ldots, q_t$ and $\lambda_1, \lambda_2, \ldots, \lambda_t$ be given and satisfy $0 \le q_1 < \lambda_1^2 < q_2 < \lambda_2^2 < \cdots < q_t < \lambda_t^2$.  Consider the linear system in $\{a_j\}$ given by
			\[
			1 = \sum_{j = 1}^t \frac{a_j}{\lambda_i^2 - q_j}, 1 \le i \le t.
			\]
			Then the solution is given by
			\[
			a_j = \frac{\prod_{i = 1}^t (\lambda_i^2 - q_j)}{\prod_{i \neq j} (q_i - q_j)}.
			\]
		\end{lemma}
		
		\begin{proof}
			The coefficient matrix for the linear system is given by the Cauchy matrix
			\[
			M = \left[ \frac{1}{\lambda_h^2 - q_i} \right]_{1 \le h, i \le t}.
			\]
			For each $j=1, \ldots, t,$ let $M_j$ denote the matrix formed from $M$ by replacing its $j$--th column with the all--ones vector. 
			
			We have that
			\[
			\det(M) = \frac{\prod_{h = 2}^t \prod_{i = 1}^{h - 1} (\lambda_h^2 - \lambda_i^2) (q_h - q_i)}{\prod_{h = 1}^t \prod_{i = 1}^t (\lambda_h^2 - q_i)}
			\]
			(see \cite{HJ}, for example) and
			\begin{align*}
				\det(M_j) &= \frac{\prod_{h = 1}^t (\lambda_h^2 - q_j)}{\prod_{i \neq j} (q_i - q_j)} \det(M) \sum_{g = 1}^t \frac{\prod_{i \neq j} (\lambda_g^2 - q_i)}{ \prod_{h \neq g} (\lambda_h^2 - \lambda_g^2)} \\
				&= \frac{\prod_{h = 1}^t (\lambda_h^2 - q_j)}{\prod_{i \neq j} (q_i - q_j)} \det(M) \sum_{g = 1}^t \frac{(-1)^{g - 1} \prod_{i \neq j} (\lambda_g^2 - q_i) \prod_{g \notin \{h, i\}} (\lambda_h^2 - \lambda_i^2)}{ \prod_{h = 2}^t \prod_{i = 1}^{h - 1} (\lambda_h^2 - \lambda_i^2)}.
			\end{align*}
			By a tedious calculation, we obtain
			\[
			\det(M_j) = \frac{\prod_{h = 1}^t (\lambda_h^2 - q_j)}{\prod_{i \neq j} (q_i - q_j)} \det(M) 
			\]
			and therefore by Cramer's Rule, we obtain
			\[
			a_j = \frac{\det(M_j)}{\det(M)} = \frac{\prod_{h = 1}^t (\lambda_h^2 - q_j)}{\prod_{i \neq j} (q_i - q_j)}
			\]
			as claimed.
		\end{proof}
		
		\begin{remark}{\rm{Lemma \ref{lem:leaf-interlace} identifies the nonzero eigenvalues in the support of the central vertex $v$. It is straightforward to determine the remaining eigenvalues of $A$ as follows, and we subdivide the discussion into two cases. \\
					{\emph{Case 1}}: $q_1 \ge 1$. In this scenario the remaining eigenvalues consist of $\pm \sqrt{q_j}$ with multiplicity $a_j-1, j=1, \ldots, t,$ as well as $0$ with multiplicity $1+\sum_{j=1}^t a_j(q_j-1)$.   We note that in this case the eigenvalue $0$ is in the support of $v.$\\
					{\emph{Case 2}}: $q_1=0$. In this scenario the remaining eigenvalues are $\pm \sqrt{q_j}$ with multiplicity $a_j-1, j=1, \ldots, t,$ as well as $0$ with multiplicity $a_1-1+\sum_{j=2}^t a_j(q_j-1)$.   In this case the eigenvalue $0$ is not in the support of $v.$
			}}    
		\end{remark}

		\section{PGST for trees of  diameter 4}

		Theorem \ref{thm:str_cospec} characterizes the three scenarios in which a pair of vertices in a tree of diameter 4 exhibits strong cospectrality. In this section, we consider each of those scenarios and analyse suitable trees in which there is PGST between the strongly cospectral vertices. 
		Specifically, 
		we construct families of trees of diameter 4 with convenient sets of eigenvalues that allow us to demonstrate PGST and in addition, to determine explicit readout times that ensure high fidelity. 
		Our results rely heavily on the machinery of Section \ref{sec:struc}. 
		
		We begin with the family with two types of stems: leaves adjacent, which corresponds to Theorem \ref{thm:str_cospec} c).
		
		\begin{theorem}\label{thm:fid_form}
			Suppose that $t=2, q_1=0, q_2=2,$ and that $k>1$ is  an odd integer. Set $a_1=k^2, a_2=k^2-1.$ Suppose further that there are sequences of positive odd integers $\mu_n, \nu_n$ such that $\lim_{n \rightarrow \infty} \frac{\mu_n}{\nu_n} = \sqrt{2},$ and let $\delta_n=\nu_n \sqrt{2}-\mu_n.$  Letting $f(\tau)$ denote the fidelity (at time $\tau$)  between leaves adjacent to a common stem, we have 
			\begin{equation}\label{eq:del_fid}
				f\left(\frac{\mu_n \pi}{\sqrt{2}}\right) = \left(  1 -\frac{1}{2(k^2-1)} \left(1-
				\cos\left(\frac{\delta_n \pi}{\sqrt{2}}\right)\right)
				\right)^2.  
			\end{equation}
			In particular, when $n$ is odd, we have 
			\begin{eqnarray}\label{eq:del_fid2} 
				&& f\left( \frac{(1+\sqrt{2})^n + (1-\sqrt{2})^n   }{2\sqrt{2}} \pi \right) 
				=  \left( 1 - \frac{1}{2k^2 - 1} \left(1 - \cos \left(\frac{(\sqrt{2} - 1)^n}{\sqrt{2}} \pi \right) \right) \right)^2.
				%\left(  \frac{1}{2} + \frac{\cos(k (\sqrt{2}-1)^n  \pi)}%{(k^2-1)%(4k^2-2)}  + \frac{k^2-2}{2(k^2-1)} \cos((\sqrt{2}-1)^n \pi)
				%\right)^2.
			\end{eqnarray}
		\end{theorem}
		\begin{proof}
			Our graph has a total of $4k^2-2$ vertices, and from the results in section 4 it follows that the eigenvalues of the adjacency matrix $A$ are $0^{(2k^2-2)}, 1, -1, \sqrt{2}k,-\sqrt{2}k,$ $ \sqrt{2}^{(k^2-2)}, -\sqrt{2}^{(k^2-2)}$ (here the superscripts denote multiplicities). Suppose for concreteness that vertices $1$ and $2$ are leaves adjacent to a common stem.
			%Let $V$ denote an orthogonal matrix that diagonalizes the adjacency matrix. 

			We have the following observations. \\
			1. $\frac{1}{\sqrt{2}}(e_1-e_2)$ is a null vector, and this can be extended to a basis of the null space in which all remaining null vectors are $0$ in the first two positions. \\
			2. For each of the eigenvalues $1$ and $-1,$ there is a corresponding $(1,-1)$ eigenvector. \\
			3. For each of the eigenvalues $\sqrt{2}k$ and $-\sqrt{2}k,$ there is a corresponding normalized eigenvector whose first entry is $\frac{1}{\sqrt{2(k^2-1)(4k^2-2)}}.$\\
			4. For each eigenvector corresponding to eigenvalue $\sqrt{2}$, we may construct a corresponding eigenvector corresponding to $-\sqrt{2}$ by changing the signs of the entries associated with all vertices on one side of the bipartition. 
			
			For each eigenvalue $\lambda$ with corresponding orthogonal eigenprojection matrix $E_\lambda,$ let $s_\lambda=(E_\lambda)_{1,1}.$ From observations 1--3, we find that $s_0=\frac{1}{2}, s_1=s_{-1}=\frac{1}{4k^2-2},$ and $ s_{\sqrt{2}k} = s_{-\sqrt{2}k}  = \frac{1}{{2(k^2-1)(4k^2-2)}}.$ From observation 4, we have $s_{\sqrt{2}} = s_{-\sqrt{2}},$ and from the fact that $e_1^T(\sum_{\lambda \in \Theta_1}E_\lambda)e_1=1$, we have $s_{\sqrt{2}} + s_{-\sqrt{2}}= 1-s_0-s_1-s_{-1}-s_{\sqrt{2}k} - s_{-\sqrt{2}k}.$ Hence,  
			$s_{\sqrt{2}} = s_{-\sqrt{2}} = \frac{k^2-2}{4(k^2-1)}.$ 
			%(Is the above sufficiently clear??) 

			It now follows that 
			\begin{eqnarray*}\label{eq:fid}
				e_1^T e^{i\tau A }e_2 &=& 
				-\frac{1}{2} + \frac{1}{4k^2-2}(e^{i\tau}+e^{-i\tau}) +\frac{1}{{2(k^2-1)(4k^2-2)}}(e^{i\tau \sqrt{2}k}+e^{-i\tau\sqrt{2}k}) \\
				&&   + \frac{k^2-2}{4(k^2-1)}(e^{i\tau \sqrt{2}}+e^{-i\tau\sqrt{2}})
				\\
				&=& -\frac{1}{2} + \frac{1}{2k^2 - 1} \cos (\tau) + \frac{1}{(k^2 - 1)(4k^2 - 2)} \cos (\tau \sqrt{2}k) + \frac{k^2 - 2}{2 (k^2 - 1)} \cos (\tau \sqrt{2}) .
			\end{eqnarray*} 
			Consider $\tau = \frac{\mu_n \pi}{\sqrt{2}}.$ 
			From our hypotheses on $\mu_n, \nu_n,$ we see that
			$\cos(\mu_n\pi)=-1, \cos(k \mu_n\pi)=-1$ and $\cos\left(\frac{\mu_n \pi}{\sqrt{2}}\right)=-\cos\left(\frac{\delta_n \pi}{\sqrt{2}}\right).$
			%$e^{i \beta_n \pi}=-1,
			%e^{i \beta_n \pi \sqrt{2}} = e^{i \delta_n \pi+ \alpha_n \pi} = -e^{i %\delta_n \pi }$  and $e^{i \beta_n \pi \sqrt{2}k} = e^{i \delta_n k %\pi + \alpha_n k  \pi} = -e^{i \delta_n k  \pi }$. 
			The expression \eqref{eq:del_fid} now follows.

			Next we consider a specific sequence of rational approximations to $\sqrt{2}.$  
			Take 
			\begin{eqnarray}\label{eq:al_be}
				\alpha_n = \frac{(1+\sqrt{2})^n + (1-\sqrt{2})^n   }{2},\quad
				\beta_n = 
				\frac{(1+\sqrt{2})^n - (1-\sqrt{2})^n   }{2\sqrt{2}},
			\end{eqnarray}
			and recall that $\frac{\alpha_n}{\beta_n}$ is the sequence of convergents associated with the continued fraction expansion for $\sqrt{2}$ (see sequences A001333 and A000129 in \cite{OEIS}, for example). In particular, $\frac{\alpha_n}{\beta_n} \rightarrow \sqrt{2} $ as $n \rightarrow \infty,$ and it is readily determined that ${\alpha_n}, {\beta_n}$ are both odd when $n$ is odd. A straightforward computation shows that $\beta_n \sqrt{2} - \alpha_n = (-1)^{n+1}(\sqrt{2}-1)^n,$ and \eqref{eq:del_fid2} now follows.

			%\textcolor{olive}{
				%\begin{align*}
				%    e_1^T e^{i \tau A} e_2 &= \frac{1}{2} - \frac{1}{2k^2 - 1} \cos (\tau) - \frac{1}{(k^2 - 1)(4k^2 - 2)} \cos (\tau \sqrt{2}k) - \frac{k^2 - 2}{2 (k^2 - 1)} \cos (\tau \sqrt{2}) \\
				%   f \left(\frac{\alpha_n \pi}{\sqrt{2}} \right) &= \left( \frac{1}{2} - \frac{1}{2k^2 - 1} \cos \left(\frac{\alpha_n \pi}{\sqrt{2}} \right) +  \frac{1}{(k^2 - 1)(4k^2 - 2)} + \frac{k^2 - 2}{2 (k^2 - 1)} \right)^2 \\
				%  &= \left( 1 - \frac{1}{2k^2 - 1} \left(1 - \cos \left(\frac{(\sqrt{2} - 1)^n}{\sqrt{2}} \pi \right) \right) \right)^2
				%\end{align*}
				%}
		\end{proof}

		\begin{example} {\rm{
					We illustrate Theorem \ref{thm:fid_form} with a small numerical example. Taking $k=3,$ the corresponding tree has $34$ vertices. Considering the convergent in \eqref{eq:al_be} with $n=3$ we have $\alpha_3=7, \beta_3=5$ and a Matlab computation yields $f(7 \pi /\sqrt{2}) \approx 0.99853\ldots.$ Taking $n=5$ we have 
					$\alpha_5=41, \beta_5=29$ and another computation gives $f(51 \pi/\sqrt{2}) \approx  0.99995\ldots.$ 
			}}
		\end{example}

		The following result also considers Theorem \ref{thm:str_cospec} c), in the case that $t=3.$ Observe that as in Theorem \ref{thm:fid_form}, the sequence of readout times $\frac{\alpha_n \pi}{\sqrt{2}},$ $n$ odd, is advantageous. 
		
		\begin{theorem}\label{thm:t=3}  Let $k_2, k_3 \in \mathbb{N}$ with $2< k_2 < k_3.$ Suppose that $T$ is a diameter $4$ tree with $t=3$ and parameters $q_1=0, q_2=2, q_3, a_1, a_2, a_3,$ and that $a, b$ are the only leaves of some stem. The eigenvalue support of $a$ (and $b$) contains the values $\pm k_2 \sqrt{2}$ and $\pm k_3\sqrt{2}$ if and only if i) $2k_2^2<q_3<2k_3^2,$ ii) $q_3|2k_2^2k_3^2$, iii)  $(q_3-2)|2(k_2^2-1)(k_3^2-1),$ and iv) $a_1=\frac{2k_2^2k_3^2}{q_3}, a_2 = \frac{2(k_2^2-1)(k_3^2-1)}{q_3-2}, a_3=(q_3-1)(a_1-a_2-1)$.

			In the event  that  i)--iv) hold and one of $k_2$ and $k_3$ is even, there is no PGST from $a$ to $b$. If i)--iv) hold and both  $k_2$ and $k_3$ are odd, then letting $f(t)$ denote the $a$ to $b$ fidelity at time $t,$ we have 
			$$f\left(\frac{\alpha_n \pi}{\sqrt{2}}\right) = 
			\left(1 - \frac{q_3-1}{(2k_2^2-1)(2k_3^2-1)} \left(1 - \cos \left( \frac{(\sqrt{2} - 1)^n}{\sqrt{2}} \pi\right) \right) \right)^2,$$ 
			where $n$ is odd and $\alpha_n$ is as in \eqref{eq:al_be}. 
		\end{theorem}
		\begin{proof}
			Observe that the eigenvalue support of $a$ consists of: $0, \pm \sqrt{2}$ and $\pm \lambda_j, j=1, \ldots, 3,$ where $0< \lambda_1 < \lambda_2 <  \lambda_3$ are the positive roots of the equation 
			\begin{equation}
				\frac{a_1}{x^2} + \frac{a_2}{x^2-2}+\frac{a_3}{x^2-q_3}=1.
				\label{eqn:eigenroots}
			\end{equation}
			From Lemma \ref{lem:Cauchy} we find that $a_1=\frac{\lambda_1^2 \lambda_2^2 \lambda_3^2}{2q_3}$, 
			$a_2=\frac{(\lambda_1^2-2)(\lambda_2^2-2)(\lambda_3^2-2)}{-2(q_3-2)}$, 
			$a_3=\frac{(\lambda_1^2-q_3)(\lambda_2^2-q_3)(\lambda_3^2-q_3)}{-q_3(2-q_3)}.$ Note also that necessarily $0< \lambda_1^2<2<\lambda_2^2<q_3<\lambda_3^2.$ 
			
			Suppose now that the support of $a$ contains $k_2 \sqrt{2}$ and $ k_3\sqrt{2};$
			then it must be the case that $\lambda_2^2=2k_2^2, \lambda_3^2=2k_3^2. $ Since $\lambda_1^2 = \frac{2a_1q_3}{4k_2^2k_3^2}\in \mathbb{Q},$ we find  that $\lambda_1^2 \in \mathbb{N}.$ As $0<\lambda_1^2<2,$ we conclude that $\lambda_1=1.$ It now follows by setting $x = 1$ in (\ref{eqn:eigenroots}) that necessarily $a_3=(q_3-1)(a_1-a_2-1).$ 
			Since $a_1=\frac{\lambda_1^2 \lambda_2^2 \lambda_3^2}{2q_3} = \frac{2k_2^2k_3^2}{q_3},$ we see that  $q_3|2k_2^2k_3^2$. Similarly, since $a_2=\frac{(\lambda_1^2-2)(\lambda_2^2-2)(\lambda_3^2-2)}{-2(q_3-2)}=\frac{2(k_2^2-1)(k_3^2-1)}{q_3-2},$ we see that $(q_3-2)|2(k_2^2-1)(k_3^2-1).$ This establishes the necessity of conditions i)--iv). 
			
			Conversely, if conditions i)--iv) hold, it is readily checked that both $ k_2 \sqrt{2}$ and $ k_3\sqrt{2}$ are solutions  to the equation $\frac{a_1}{x^2} + \frac{a_2}{x^2-2}+\frac{a_3}{x^2-q_3}=1.$

			Suppose henceforth that i)--iv) hold. Suppose that one of $k_2, k_3$ is even, and for concreteness we assume that in fact $k_2$ is even. Note that $(k_2-1)\times 0+ (-k_2)\times \sqrt{2} + (1)\times k_2\sqrt{2}=0, k_2-1$ is odd, and $(k_2-1)+ (-k_2) + (1)=0.$ From Theorem 2 of \cite{BCGS}, we deduce that there cannot be PGST from $a$ to $b$. 
			
			Now suppose that both $k_2$ and $k_3$ are odd. The eigenvalues in the support of $a$ and $b$ are: $0, \pm 1, \pm \sqrt{2}, \pm k_2 \sqrt{2}$ and $\pm k_3 \sqrt{2} $. Note that $\sigma^-=\{0\}.$ Then $(E_0)_{a,b}=-\frac{1}{2}, (E_{\pm 1})_{a,b}= \frac{q_3-1}{2(2k_2^2-1)(2k_3^2-1)}$ and 
			\begin{eqnarray*}
				&&(E_{\sqrt{2}})_{a,b} + (E_{-\sqrt{2}})_{a,b}+(E_{k_2\sqrt{2}})_{a,b} + (E_{-k_2\sqrt{2}})_{a,b}+(E_{k_3\sqrt{2}})_{a,b} + (E_{-k_3\sqrt{2}})_{a,b}\\
				&& = \frac{1}{2}-(E_{ 1})_{a,b} -(E_{-1})_{a,b}.
			\end{eqnarray*}
			It now follows that   
			\begin{eqnarray*}
				&&f\left(\frac{\alpha_n \pi}{\sqrt{2}}\right) =
				\left(1 - 2(E_1)_{a,b} \left(1 - \cos \left( \frac{(\sqrt{2} - 1)^n}{\sqrt{2}} \pi\right) \right) \right)^2 = \\
				&&\left(1 - \frac{q_3-1}{(2k_2^2-1)(2k_3^2-1)} \left(1 - \cos \left( \frac{(\sqrt{2} - 1)^n}{\sqrt{2}} \pi\right) \right) \right)^2.  
			\end{eqnarray*}
			
		\end{proof}
		
		\begin{example}\label{eg:t=3}
			{\rm{Here is an example that illustrates Theorem \ref{thm:t=3}. 
					Suppose that $t = 3$, $a_1 = 99$, $q_1 = 0$, $a_2 = 96$, $q_2 = 2$, $a_3 = 42$, $q_3 = 22$.  This graph has a total of 1354 vertices, and the eigenvalues of the adjacency matrix (with corresponding multiplicities) are $0^{(1076)},$ $\pm 1,$ $ \pm \sqrt{2}^{(95)},$ $ \pm 3 \sqrt{2},$ $ \pm \sqrt{22}^{(41)},$ $ \pm 11 \sqrt{2}$.  If $a$ and $b$ are the only leaves of some stem, then we have
					\begin{comment}
						\begin{eqnarray*}
							&&  (E_0)_{a, b} = - \frac{1}{2};     (E_{\pm 1})_{a, b} = \frac{21}{8194};     (E_{\pm \sqrt{2}})_{a, b} = \frac{95}{384}; \\
							&& (E_{\pm 3 \sqrt{2}})_{a, b} = \frac{1}{30464}; 
							(E_{\pm \sqrt{22}})_{a, b} = 0; 
							(E_{\pm 11 \sqrt{2}})_{a, b} = \frac{11}{1295616}.
						\end{eqnarray*}
						It now follows that 
					\end{comment}
					\begin{align*}
						U(\tau)_{a, b} &= - \frac{1}{2} + \frac{21}{4097} \cos (\tau) + \frac{95}{192} \cos (\sqrt{2} \tau) + \frac{1}{15232} \cos (3 \sqrt{2} \tau) +\frac{11}{647808} \cos (11 \sqrt{2} \tau).
					\end{align*} 
					In particular for $n$ odd and $\alpha_n$ given by \eqref{eq:al_be}, we have 
					$$
					f \left(\frac{\alpha_n \pi}{\sqrt{2}} \right) = \left(1 - \frac{21}{4097} \left(1 - \cos \left( \frac{(\sqrt{2} - 1)^n}{\sqrt{2}} \pi\right) \right) \right)^2. 
					$$
					For instance, $f \left( \frac{7 \pi}{\sqrt{2}} \right) \approx 0.999873 \ldots$ and $
					f \left( \frac{41 \pi}{\sqrt{2}} \right) \approx 0.999996 \ldots .$
			}}
		\end{example}
		
		Our next result considers the scenario of Theorem \ref{thm:str_cospec} a). 
		
		\begin{theorem} \label{thm:qreadout}
			Suppose that $q \in \mathbb{N}$ and consider the diameter 4 tree with $t=1, q_1=q, a_1=2.$ For each $\ell \in \mathbb{N}\cup \{0\},$ define 
			$$D_\ell = \frac{r_1^{2\ell}(2q+1-r_2)(r_1-1) -r_2^{2\ell}(2q+1-r_1)(r_2-1)}{2(r_1-r_2)}, $$ where $r_1=q+1+\sqrt{q^2+2q}, r_2=q+1-\sqrt{q^2+2q}.$ Set $\delta_\ell = \frac{r_2^{2\ell}(r_2-1)}{2}\left(\sqrt{\frac{q+2}{q}} - 1  \right)$ and $\tau_\ell = \frac{D_\ell \pi}{\sqrt{q}}.$ 
			
			The stem-to-stem fidelity at time $\tau_\ell$ is $\frac{1}{4}(1+ \cos(\delta_\ell \pi))^2.$
		\end{theorem}
		\begin{proof}
			%   {\emph{SKETCH}}\\
			
			For the stems $a,b$ we have $\sigma_{ab}^+=\{\pm \sqrt{q+2}\}, \sigma_{ab}^-=\{\pm \sqrt{q}\}$ and for each of the corresponding unit eigenvectors, the entries corresponding to the stems are $\frac{1}{2}$ or $-\frac{1}{2}.$ It follows that at any time $\tau,$ the stem-to-stem fidelity is $\frac{1}{4}(-\cos(\tau \sqrt{q}) + \cos(\tau \sqrt{q+2}) )^2.  $ 
			
			Define 
			$$N_\ell = \frac{r_1^{2\ell}(2q+3-r_2)(r_1-1) -r_2^{2\ell}(2q+3-r_1)(r_2-1)}{2(r_1-r_2)}.$$ 
			We claim that if $q$ is odd, then for each $\ell \in \mathbb{N},$ $N_\ell$ and $D_\ell$ are even and odd natural numbers, respectively. (It turns out that $\frac{N_\ell}{D_\ell}$ is a convergent arising from the continued fraction expansion for $\sqrt{\frac{q+2}{q}}$, hence our interest in $N_\ell$ and $D_\ell$.)

			To see the claim, set $x(1)=\begin{bmatrix}
				q+1\\1
			\end{bmatrix}, y(1)=\begin{bmatrix}
				q\\1
			\end{bmatrix}$ and define the recurrences $x(j+1)=\begin{bmatrix}
				2q+1 &q\\2 &1
			\end{bmatrix}x(j)$ and $ y(j+1)=\begin{bmatrix}
				2q+1 &q\\2 &1
			\end{bmatrix}y(j), j \in \mathbb{N}.$ Evidently $x(j), y(j) \in \mathbb{N}^2$ for all  $j \in \mathbb{N},$ and a straightforward induction proof shows that whenever $j$ is odd, $e_1^Tx(j)$ is even, while $e_2^Tx(j), e_1^Ty(j), e_2^Ty(j)$  are odd. The eigenvalues of $\begin{bmatrix}
				2q+1 &q\\2 &1
			\end{bmatrix}$ are readily seen to be $r_1, r_2,$ with corresponding eigenvectors $v_1=\begin{bmatrix}
				r_1-1\\2 
			\end{bmatrix}, v_2=\begin{bmatrix}
				r_2-1\\2 
			\end{bmatrix},$ respectively. 
			
			We can write $$x(1) = 
			\frac{(2q+3-r_2)(r_1-1)}{2(r_1-r_2)}v_1 
			-\frac{(2q+3-r_1)(r_2-1)}{2(r_1-r_2)}v_2$$ and $$y(1) =  \frac{(2q+1-r_2)(r_1-1)}{2(r_1-r_2)}v_1 
			-\frac{(2q+1-r_1)(r_2-1)}{2(r_1-r_2)}v_2$$ so as to deduce that 
			$$x(2\ell+1) = 
			\frac{r_1^{2\ell}(2q+3-r_2)(r_1-1)}{2(r_1-r_2)}v_1 
			-\frac{r_2^{2\ell}(2q+3-r_1)(r_2-1)}{2(r_1-r_2)}v_2$$ and 
			$$y(2\ell+1) = 
			\frac{r_1^{2\ell}(2q+1-r_2)(r_1-1)}{2(r_1-r_2)}v_1 
			-\frac{r_2^{2\ell}(2q+1-r_1)(r_2-1)}{2(r_1-r_2)}v_2.$$ In particular we find that for each $\ell \in \mathbb{N}\cup \{0\},$ $N_\ell= e_1^Tx(2\ell+1)$ is an even natural number, while $D_\ell= e_1^Ty(2\ell+1)$ is an odd natural number. 
			This completes the proof of the claim. 
			
			In the case that $q$ is even, an analogous argument shows that for each $\ell \in \mathbb{N}, N_\ell$ and $D_\ell$ are odd and even, respectively. 
			
			A computation shows that $ D_\ell \sqrt{\frac{q+2}{q}} - N_\ell =\delta_\ell .$ Hence, when $q$ is odd, 
			at time $\tau_\ell$ we have $\cos(\tau_\ell \sqrt{q})=\cos(D_\ell \pi)=-1$ and 
			$\cos(\tau_\ell \sqrt{q+2}) = \cos(D_\ell \sqrt{\frac{q+2}{q}} \pi) = \cos((N_\ell + \delta_\ell)\pi) = \cos(\delta_\ell\pi). $ Similarly, when $q$ is even, 
			we have $\cos(\tau_\ell \sqrt{q})=\cos(D_\ell \pi)=1$ and 
			$\cos(\tau_\ell \sqrt{q+2}) = \cos(D_\ell \sqrt{\frac{q+2}{q}} \pi) = \cos((N_\ell + \delta_\ell)\pi) = -\cos(\delta_\ell\pi). $

			The conclusion now follows. 
		\end{proof}
		
		\begin{remark}\label{rm:p5} {\rm {Observe that in Theorem \ref{thm:qreadout}, the case $q=1$ yields the graph $P_5,$ which exhibits PGST between next-to-leaf vertices.  Thus Theorem  \ref{thm:qreadout} provides an advantageous sequence of readout times for a family of trees that includes $P_5$. 
					
					An analysis similar to that in Theorem \ref{thm:qreadout} also shows that for $P_5$ the fidelity between end vertices at time $D_\ell \pi$ is equal to $\left( \frac{5}{6} + \frac{\cos(\delta_\ell \pi)}{6}\right)^2.$ Hence that sequence of readout times also ensures good fidelity between the leaves of $P_5.$ 
			}}
		\end{remark}

		The following technical result will be useful in the proof of Theorem \ref{thm:pend-dist-4} below. 
		
		\begin{lemma}\label{lem:fid}
			Suppose that $G$ is a graph with strongly cospectral vertices $a,b$ and denote the eigenvalues in $\sigma_{ab}$ by $\lambda_j, j=1, \ldots, m.$ Let $0<\epsilon < \frac{1}{2}$ be given, and suppose that for some $\tau>0$ there are integers $n_j, j=1, \ldots, m$ such that 
			$|\tau\lambda_j-n_j|< \frac{\epsilon}{2\pi}$ for each 
			$\lambda_j \in \sigma_{ab}^+$ and 
			$|\tau\lambda_j-n_j-\frac{1}{2}|< \frac{\epsilon}{2\pi}$ for each $\lambda_j \in \sigma_{ab}^-.$ 
			Then the $a$ to $b$ fidelity at time $2\pi \tau$ is bounded below by $1-2\epsilon.$
		\end{lemma}
		\begin{proof} 
			Since $a$ and $b$ are strongly cospectral, it follows that there are positive numbers $c_j, j=1,\ldots, m$ with $\sum_{j=1}^mc_j=1$ such that the fidelity $f$ at time $2\pi \tau$ can be written as 
			\begin{align*}
				f &=\left|\sum_{\lambda_j \in \sigma_{ab}^+}c_je^{2\pi \tau i \lambda_j} - \sum_{\lambda_j \in \sigma_{ab}^-}c_je^{2\pi \tau i \lambda_j} \right|^2 \\
				&=\left|\sum_{\lambda_j \in \sigma_{ab}^+}c_j + \sum_{\lambda_j \in \sigma_{ab}^-}c_j + \sum_{\lambda_j \in \sigma_{ab}^+}(c_je^{2\pi \tau i \lambda_j}-1)   - \sum_{\lambda_j \in \sigma_{ab}^-}c_j(e^{2\pi \tau i \lambda_j}+1) \right|^2 \\
				&\ge \left| 1- \sum_{\lambda_j \in \sigma_{ab}^+}c_j \left| e^{2\pi \tau i \lambda_j}-1 \right| - \sum_{\lambda_j \in \sigma_{ab}^-}c_j \left|e^{2\pi \tau i \lambda_j}+1 \right| \right|^2,
			\end{align*}
			since  $\sum_{j=1}^mc_j=1$ together with an application of the triangle inequality. If $\lambda_j \in \sigma_{ab}^+ $ we have 
			$$|e^{2\pi \tau i \lambda_j}-1| = \sqrt{2(1-\cos(2\pi \tau\lambda_j))}
			\le  \sqrt{2(1-\cos(\epsilon))} \le \epsilon.$$ We find similarly that if $\lambda_j \in \sigma_{ab}^- $ then 
			$|e^{2\pi \tau i \lambda_j}+1| \le \epsilon.$ The conclusion now follows. 
		\end{proof}
		
		We next attempt to localize an advantageous readout time in the context of Theorem \ref{thm:str_cospec} b). 
		
		\begin{theorem}\label{thm:pend-dist-4} Suppose that $t=2,$ $a_1=2, q_1=1,$ and let $a,b$ denote  the leaves on the branches at $v$ corresponding to $q_1.$ Suppose 
			that for each eigenvalue $\lambda$ in the support of $a$ we have $\lambda^2 \in \mathbb{N}.$ Then $q_2> 2, $ and the eigenvalues in the support of $a$ are $\pm 1, \pm \sqrt{2}, \pm \sqrt{2q_2-1}.$  
			In particular, there is PGST between $a$ and $b$ if and only if $2q_2-1$ is not a perfect square. 
			
			Suppose $2q_2-1$ is not a perfect square and that $0<\epsilon < \frac{1}{2}.$ There is a $\tau_0$ such that $|\tau_0|\le \frac{2^{24}3^{11}\pi^4(2q_2-1)^7}{\epsilon^3}+\pi$ such that the $a$ to $b$ fidelity is bounded below by $1-2\epsilon.$ 
		\end{theorem}
		\begin{proof}
			Observe that $\pm 1 \in \sigma_{ab}^-.$ 
			Suppose that the branches at $v$ not containing $a$ or $b$ have stems of degree $q_2$. Then $\sigma_{ab}^+$ consists of $0,$ along with the roots of the following equation: $\frac{2}{x^2-1}+\frac{a_2}{x^2-q_2}=1.$ Denote those roots by $\pm \lambda_1, \pm \lambda_2$ where $0<\lambda_1 < \lambda_2, $ and observe that if $q_2=0,$ then  $0<\lambda_1<1,$ contrary to our hypothesis. Hence, $q_2 \ge 2.$ 
			
			We have $2=\frac{(\lambda_1^2-1)(\lambda_2^2-1)}{q_2-1}$ and since $1<\lambda_1^2<q_2<\lambda_2^2,$ we find that $2>\lambda_1^2-1.$ As $\lambda_1^2 \in \mathbb{N},$ it must be the case that $\lambda_1^2=2.$ Substituting $x^2=2$ into the equation  $\frac{2}{x^2-1}+\frac{a_2}{x^2-q_2}=1,$ we find that $a_2=q_2-2.$ It is now straightforward to determine that $x^2=2q_2-1$ yields the other root of the equation  $\frac{2}{x^2-1}+\frac{q_2-2}{x^2-q_2}=1.$ Hence $\sigma_{ab}^+=\{  \pm \sqrt{2}, \pm \sqrt{2q_2-1}\}.  $
			
			If $ {2q_2-1}$ is not a perfect square, we find that  $\pm 1, \pm \sqrt{2}, \pm \sqrt{2q_2-1}$ are independent over $\mathbb{Q},$ which ensures that there is PGST between $a$ and $b$. Next, suppose that  ${2q_2-1}$ is a perfect square, say with $ \sqrt{2q_2-1}=k\in \mathbb{N}.$ Consider the following integer parameters: $\ell_1=-k, \ell_1'=0, \ell_{\sqrt{2}}=\frac{k-1}{2}, \ell_{\sqrt{2}}'=\frac{k-1}{2}, \ell_k=1, \ell_k'=0, \ell_0=0. $ Then the following equations are satisfied: 
			\begin{eqnarray*}
				&& (1)\ell_1 + (-1)\ell_1'+ ({\sqrt{2}})\ell_{\sqrt{2}} + (-{\sqrt{2}})\ell_{\sqrt{2}}' + (k)\ell_k+(-k)\ell_k' + (0)\ell_0=0,\\ 
				&& \ell_1+\ell_1'=-k \ {\rm{(odd)}}, \\
				&& \ell_1 + \ell_1'+ \ell_{\sqrt{2}} + \ell_{\sqrt{2}}' + \ell_k+\ell_k'+\ell_0=0.
			\end{eqnarray*}
			By Theorem 2 of \cite{BCGS}, PGST between $a$ and $b$ cannot hold. 
			
			%\textcolor{red}{This next bit needs repair.}
			Next we suppose that $2q_2-1$ is not a perfect square, and that $0<\epsilon<\frac{1}{2}.$ According to Theorem 3.1 of \cite{FM}, there is an integer $r$ and integers $n_1, n_2$ such that $|r\sqrt{2}+\frac{\sqrt{2}}{2}-n_1| \le \frac{\epsilon}{2\pi}, |r\sqrt{2q_2-1}+\frac{\sqrt{2q_2-1}}{2}-n_2| \le \frac{\epsilon}{2\pi}.$  
			Evidently for such an $r$ we have $|(2r+1)\pi \sqrt{2} -2\pi n_1|<\epsilon,  |(2r+1)\pi \sqrt{2q_2-1} -2\pi n_1|<\epsilon .$
			Further, that result  of \cite{FM}  provides an upper bound on $|r|$ that depends on: the degrees of the minimal polynomials of $\sqrt{2q_2-1}$ and $\sqrt{2},$ the Weil heights (see \cite{Z}) of $\sqrt{2q_2-1}$ and $\sqrt{2},$ and the dimension of $\mathbb{Q}(\sqrt{2q_2-1}, \sqrt{2})$ over $\mathbb{Q}$. Computing those quantities we find that the minimal polynomials have degree $2,$ that the Weil heights of $\sqrt{2q_2-1}$ and $\sqrt{2}$ are  $\sqrt{2q_2-1}$ and $\sqrt{2},$ respectively, and that the dimension of $\mathbb{Q}(\sqrt{2q_2-1}, \sqrt{2})$ over $\mathbb{Q}$ is $4.$
			Using that information and then applying the upper bound of \cite[Theorem 3.1]{FM}, we find that 
			$$|r|\le 2^{-8}6^8(3(\sqrt{2q_2-1})^2)^3(2\sqrt{2})^8(2\sqrt{2q_2-1})^8\left(\frac{\epsilon}{2\pi}\right)^{-3}=\frac{2^{23}3^{11}\pi^3(2q_2-1)^7}{\epsilon^3}.
			$$
			Setting $\tau_0$ equal to $\pi(2 r+1)$ now yields the desired conclusion.   
		\end{proof}

		\begin{remark}{\rm {
					The coefficient ${2^{24}3^{11}\pi^4}$ in Theorem \ref{thm:pend-dist-4} is on the order of $10^{14
					}.$}}
		\end{remark}
		
		Evidently the readout time bound of Theorem \ref{thm:pend-dist-4} is quite large. The next example shows that under certain circumstances, a more precise statement can be made about readout times corresponding to good fidelity. 
		
		\begin{example}{\rm{
					Here we consider the family of trees dealt with in Theorem \ref{thm:pend-dist-4}, and maintain the notation of that result. Note that the bound on $\tau_0$ in Theorem \ref{thm:pend-dist-4} does little to inform us on the value of a suitable $\tau_0$. In this example we show how $q_2$  and the readout time can be  chosen together in such a way as to provide a sequence of trees and readout times for which there is good fidelity of state transfer between the leaves $a,b$ that are at distance $4$. For this family of trees we find the following entries in the relevant eigenprojection matrices: $(E_{\pm 1})_{a,a}=\frac{1}{4},(E_{\pm \sqrt{2}})_{a,a}=\frac{q_2-2}{8q_2-12}, (E_{\pm \sqrt{2q_2-1}})_{a,a}=\frac{1}{4(4q_2^2-8q_2+3)}, (E_{0})_{a,a}=\frac{q_2}{4q_2-2}. $
					
					Suppose that $\alpha_n, \beta_n$ are as in \eqref{eq:al_be}. It is straightforward to show (for example by induction) that $\alpha_n$ is odd for all $n\in \mathbb{N},$ while $\beta_n$ is odd or even according as $n$ is odd or even. Observe that $(\beta_n+\beta_{n+1})\sqrt{2} =\alpha_n+\alpha_{n+1} + (-1)^{n+1}(2-\sqrt{2})(\sqrt{2}-1)^n. $ Now select $q_2 = \frac{(\beta_n+\beta_{n+1})^2+3}{2},$ so that $2q_2-1=(\beta_n+\beta_{n+1})^2+2.$ We find that 
					\begin{eqnarray*}
						&&(\beta_n+\beta_{n+1}) \sqrt{2q_2-1} = (\beta_n+\beta_{n+1})\sqrt{(\beta_n+\beta_{n+1})^2+2} \\
						&&= (\beta_n+\beta_{n+1})^2+1- \frac{1}{(\beta_n+\beta_{n+1})\sqrt{(\beta_n+\beta_{n+1})^2+2} +  (\beta_n+\beta_{n+1})^2+1}.
					\end{eqnarray*}
					
					Set $t_0=(\beta_n+\beta_{n+1})\pi . $
					Letting $f(t)$ denote the fidelity of $ab$ state transfer at time $t,$ we have 
					\begin{eqnarray*}
						&&f(t_0)=\\
						&&\Bigg|-\frac{1}{2} \cos(t_0) + \frac{q_2}{4q_2-2} + \frac{q_2-2}{4q_2-6}\cos(\sqrt{2}t_0) +\frac{1}{(4q_2-2)(2q_2-3)} \cos(\sqrt{2q_2-1}t_0) \Bigg|^2\\
						&&=\Bigg|1 - \frac{q_2-2}{4q_2-6}(1-\cos(\sqrt{2}t_0)) - \frac{1}{(4q_2-2)(2q_2-3)}(1- \cos(\sqrt{2q_2-1}t_0) )\Bigg|^2.
					\end{eqnarray*}
					
					From our observations above, we find that 
					$$\cos(\sqrt{2}t_0) = \cos((2-\sqrt{2})\pi(\sqrt{2}-1)^n), $$ 
					\begin{eqnarray*}
						&&  \cos(\sqrt{2q_2-1}t_0) =     \cos \left(  \frac{\pi}{(\beta_n+\beta_{n+1})\sqrt{(\beta_n+\beta_{n+1})^2+2} +  (\beta_n+\beta_{n+1})^2+1}\right).       
					\end{eqnarray*}
					
					For example, taking $n=2$ we have $t_0=7 \pi,$ $q_2 = 26$ and $f(t_0)\approx 0.9759 \ldots .$ With $n=3$ we get $t_0=17 \pi,$ $q_2 = 146$ and $f(t_0)\approx 0.9979 \ldots .$ 
			} } 
		\end{example}

		Our final three examples serve to quantify Corollary \ref{cor:sens} by providing the value of the sensitivity of the fidelity in some of the families of trees studied in this section. 
		
		\begin{example}\label{ex:senseeg1}
			{\rm{Here we compute the sensitivity of the fidelity with respect to the readout time in the setting of Theorem \ref{thm:qreadout}. Maintaining the notation of that result, we find that $\cos(\tau_\ell \sqrt{q}) = (-1)^q,  \sin(\tau_\ell \sqrt{q}) =0, \cos(\tau_\ell \sqrt{q+2}) = (-1)^{q-1}\cos(\delta_\ell \pi), \sin(\tau_\ell \sqrt{q+2}) = (-1)^{q-1}\sin(\delta_\ell \pi).$ Denoting the stems by $a, b,$ we have $\sigma_{ab}^+=\{\pm \sqrt{q+2}\}, \sigma_{ab}^-=\{\pm \sqrt{q}\},$ and $s_\lambda=\frac{1}{4}$ for each eigenvalue $\lambda$ in the support of $a$.    
					
					Appealing to \eqref{eq:sens}, it follows that 
					\begin{eqnarray*}
						&&  \frac{df}{dt}\Big|_{\tau_\ell} = 
						\frac{1}{8}
						\{-2\sqrt{q+2}\sin(2\tau_\ell\sqrt{q+2}) -2\sqrt{q}\sin(2\tau_\ell\sqrt{q}) + \\
						&&2(\sqrt{q+2}-\sqrt{q})\sin((\sqrt{q+2}-\sqrt{q})\tau_\ell) + 
						2(\sqrt{q+2}+\sqrt{q})\sin((\sqrt{q+2}+\sqrt{q})\tau_\ell)\}. 
					\end{eqnarray*}
					Substituting the various sine and cosine values, we find that 
					\begin{eqnarray*}
						\frac{df}{dt}\Big|_{\tau_\ell} &=& 
						\frac{1}{8} \{-2\sqrt{q+2} \sin(2 \delta_\ell \pi) -4\sqrt{q+2} \sin( \delta_\ell \pi) \} =
						-\frac{ \sin( \delta_\ell \pi)(1+ \cos( \delta_\ell \pi))\sqrt{q+2} }{2}. 
					\end{eqnarray*}

			}}    
		\end{example}
		
		\begin{example}\label{ex:senseeg2}{\rm{
					We now consider the sensitivity of the end-to-end fidelity for $P_5$ (see Remark \ref{rm:p5} above). For end vertices $a,b$ we have $\sigma_{ab}^+=\{0, \pm \sqrt{3}\}, \sigma_{ab}^-=\{\pm 1\},$ and adopting the notation of Theorem \ref{thm:sens}, $s_0=\frac{1}{3}, s_{\pm \sqrt{3}}=\frac{1}{12}, s_{\pm 1}=\frac{1}{4}.$ Using \eqref{eq:sens} and the values of the $\sin $ and $\cos$ functions at $\tau_\ell $ and $\sqrt{3}\tau_\ell$ found in Example \ref{ex:senseeg1}, a computation reveals that the sensitivity of the end-to-end fidelity at time $\tau_\ell$ is given by $-\frac{\sqrt{3} \sin(\delta_\ell \pi)(5 +\cos(\delta_\ell \pi) )}{18}.$
					
			}}
		\end{example}
		
		\begin{example}\label{ex:senseeg3}{\rm{
					Here we revisit Theorem \ref{thm:fid_form}, and compute the sensitivity of state transfer between leaves of the same stem at time $\tau = \frac{\alpha_n \pi}{\sqrt{2}}$ when $n$ is odd. Denote the two leaves by $a,b$, and observe that $\sigma_{ab}^-=\{0\}, $ while $\sigma_{ab}^+=\{\pm 1, \pm \sqrt{2}, \pm \sqrt{2k}\}.$ We have $s_0=\frac{1}{2}, s_{\pm 1}=\frac{1}{4k-2}, s_{\pm \sqrt{2}}= \frac{k-2}{4(k-1)}, s_{\pm \sqrt{2k}} = \frac{1}{2(k-1)(4k-2)}. $
					
					Set $\gamma_n= \frac{(\sqrt{2}-1)^n}{\sqrt{2}},$ so that $\frac{\alpha_n \pi}{\sqrt{2}} =\beta_n \pi -\gamma_n \pi$ (recall that $\beta_n$ is odd). It now follows that $\sin(\sqrt{2}\tau)=0, \cos(\sqrt{2}\tau)=-1, \sin(\sqrt{2k}\tau)=0, \cos(\sqrt{2k}\tau)=-1, \sin(\tau)=\sin(\gamma_n \pi), \cos(\tau)=-\cos(\gamma_n \pi).  $ Referring to \eqref{eq:sens}, we see that the first summation inside the brackets contributes $2s_0s_1 \sin(\gamma_n \pi),$ while 
					the 
					second summation is zero. A long computation (in which many terms cancel or are zero) yields the fact that the third summation is equal to $2s_1^2 \sin(2\gamma_n \pi) + 4s_1s_{\sqrt{2}}\sin(\gamma_n \pi)+4s_1s_{\sqrt{2k}}\sin(\gamma_n \pi).  $
					Substituting the three summations into \eqref{eq:sens}, we then find that the sensitivity of the fidelity at time $\tau$ is given by 
					$$
					4s_1(s_0 \sin(\gamma_n \pi) + s_1\sin(2\gamma_n \pi)+2s_{\sqrt{2}}\sin(\gamma_n \pi) + s_{\sqrt{2k}}\sin(\gamma_n \pi)),
					$$
					which in turn can be written as 
					$$\frac{4 \sin(\gamma_n \pi)}{4k-2} \left(\frac{1}{2} + \frac{2 \cos(\gamma_n \pi)}{4k-2} + \frac{k-2}{2(k-1)} + \frac{2}{2(k-1)(4k-2)}.
					\right) 
					$$
					This last expression can be simplified to yield $$\frac{2\sin(\gamma_n \pi)(2k-2+\cos(\gamma_n \pi))}{(2k-1)^2}.
					$$
					
			}}    
		\end{example}
		
		In Examples \ref{ex:senseeg1}--\ref{ex:senseeg3} we see 
		that the readout times do not necessarily correspond to the optimal fidelity (i.e. local maxima)  in their respective neighbourhoods. Nevertheless, each of those sequences of  readout times provides a convenient, systematic approach to achieving high fidelity state transfer, suffers little impact from minor perturbations, and approaches optimal fidelity.\\
		
		\noindent 
		\emph{Acknowledgement}: The authors are grateful to an anonymous referee, whose comments helped to improve the presentation of the results.
		% Strong cospectral result c.f. CanaDAM - Done
		% Readout times
		% Start by calculating eigenvectors
		% Kronecker: dimension type semi-constructive argument?

	\end{document}